\documentclass[11pt, letter]{article}
\usepackage[margin=1in]{geometry}

\usepackage{microtype}
\usepackage{graphicx}
\usepackage{subfigure}
\usepackage{booktabs} 
\usepackage{amsmath}
\usepackage{amsthm}
\usepackage{amssymb}
\usepackage{wrapfig}
\usepackage{graphicx}
\usepackage[dvipsnames]{xcolor}  
\usepackage{multirow}
\usepackage{epsfig,amssymb,amsfonts,amsmath,amsthm} 
\usepackage{appendix}
\usepackage[ruled,linesnumbered]{algorithm2e}
\usepackage{hhline}  
\usepackage{times}
\usepackage{placeins} 

\usepackage{xcolor}
\usepackage[shortlabels]{enumitem}

\definecolor{newblue}{rgb}{0.2,0.2,0.6} 

\usepackage[pagebackref=false,colorlinks=false,bookmarks=false]{hyperref}

\newcommand{\fullexponent}{\bigo{\frac{d}{\epsilon^2} \log(d) \max\left(\log\left(\frac{1}{c\sqrt{\alpha}}\right), \log\left(\frac{1}{(1/2) - c \sqrt{\alpha}}\right)\right)}}
\newcommand{\exponentsmallc}{\bigo{\frac{d}{\epsilon^2} \log(d) \log\left(\frac{1}{c\sqrt{\alpha}}\right)}}
\newcommand{\algname}{\textsc{Randomised-KS}}

\usepackage{tikzit}

\tikzstyle{basic}=[fill=white, draw=black, shape=circle]
\tikzstyle{square}=[fill=white, draw=black, shape=rectangle]
\tikzstyle{big dashed}=[fill=white, draw=black, shape=circle, minimum width=1cm, dashed]
\tikzstyle{vertical ellipse dashed}=[fill=none, draw=blue, minimum width=0.75cm, minimum height=3cm, ellipse, dashed, tikzit shape=rectangle, tikzit draw=blue, tikzit fill=white]
\tikzstyle{small vertical ellipse dashed}=[fill=none, draw=blue, shape=circle, tikzit fill=white, tikzit draw=blue, dashed, minimum width=0.75cm, minimum height=1.5cm, tikzit shape=rectangle, ellipse]
\tikzstyle{tiny vertical ellipse dashed}=[fill=none, draw=blue, shape=circle, tikzit fill=white, ellipse, dashed, minimum width=0.75cm, minimum height=1cm, tikzit shape=rectangle]
\tikzstyle{green}=[fill={rgb,255: red,0; green,128; blue,128}, draw=black, shape=circle]
\tikzstyle{red}=[fill=red, draw=black, shape=circle]
\tikzstyle{blue}=[fill=blue, draw=black, shape=circle]
\tikzstyle{orange}=[fill=orange, draw=black, shape=circle]
\tikzstyle{green}=[fill=green, draw=black, shape=circle]
\tikzstyle{brown}=[fill=brown, draw=black, shape=circle]
\tikzstyle{green}=[fill=green, draw=black, shape=circle]
\tikzstyle{magenta}=[fill=magenta, draw=black, shape=circle]
\tikzstyle{cyan}=[fill=cyan, draw=black, shape=circle]
\tikzstyle{yellow}=[fill=yellow, draw=black, shape=circle]
\tikzstyle{pink}=[fill=pink, draw=black, shape=circle]
\tikzstyle{teal}=[fill=teal, draw=black, shape=circle]
\tikzstyle{huge dashed}=[fill=white, draw=black, shape=circle, dashed, minimum width=2cm]
\tikzstyle{medium}=[fill=white, draw=black, shape=circle, minimum width=1cm]
\tikzstyle{pale green}=[fill={rgb,255: red,173; green,231; blue,0}, draw=black, shape=circle, minimum width=1cm]
\tikzstyle{horizontal ellipse dashed}=[fill=white, draw=black, tikzit draw=magenta, tikzit shape=rectangle, minimum width=3cm, minimum height=0.75cm, ellipse, dashed]
\tikzstyle{minsize}=[fill=white, draw=black, shape=circle, minimum width=0.75cm]
\tikzstyle{horizontal ellipse green}=[fill={rgb,255: red,191; green,255; blue,0}, draw=black, tikzit draw={rgb,255: red,191; green,255; blue,0}, tikzit shape=rectangle, minimum width=3cm, minimum height=0.75cm, ellipse, dashed]
\tikzstyle{horizontal ellipse blue}=[fill={rgb,255: red,107; green,203; blue,255}, draw=black, tikzit draw=blue, tikzit shape=rectangle, minimum width=3cm, minimum height=0.75cm, ellipse, dashed]
\tikzstyle{smallblack}=[fill=black, draw=black, shape=circle, inner sep=0 pt, minimum size=3 pt]
\tikzstyle{smallSquare}=[fill=white, draw=black, shape=rectangle, inner sep=0 pt, minimum size=6 pt]
\tikzstyle{smallCircle}=[fill=white, draw=black, shape=circle, inner sep=0 pt, minimum size=6 pt]
\tikzstyle{big vertical ellipse dashed}=[fill=none, draw=blue, shape=circle, tikzit shape=rectangle, ellipse, dashed, minimum width=0.95cm, minimum height=3.7cm]
\tikzstyle{smallred}=[fill=red, draw=red, shape=circle, inner sep=0 pt, minimum size=3 pt]

\tikzstyle{directed}=[->]
\tikzstyle{undirected}=[-, line width=1pt]
\tikzstyle{directed red}=[draw={rgb,255: red,255; green,55; blue,55}, ->, line width=2pt]
\tikzstyle{directed green}=[draw={rgb,255: red,0; green,128; blue,128}, ->, line width=1pt]
\tikzstyle{directed blue}=[draw={rgb,255: red,55; green,55; blue,255}, ->, line width=2pt]
\tikzstyle{directed purple}=[draw={rgb,255: red,128; green,0; blue,128}, ->, line width=1pt]
\tikzstyle{undirected red}=[-, draw=red, line width=1pt]
\tikzstyle{undirected green}=[-, draw={rgb,255: red,0; green,107; blue,61}, line width=1pt]
\tikzstyle{undirected blue}=[-, draw=blue, line width=1pt]
\tikzstyle{filled set cell}=[-, fill={rgb,255: red,175; green,255; blue,197}, tikzit fill={rgb,255: red,149; green,255; blue,179}]
\tikzstyle{undirected dashed}=[-, line width=1pt, dashed]
\tikzstyle{orange dashed}=[-, draw={rgb,255: red,255; green,128; blue,0}, dashed, line width=1.5pt]
\tikzstyle{directed dash}=[->, dashed]
\tikzstyle{blue dashed}=[-, draw=blue, dashed, line width=1pt]
\tikzstyle{green dashed}=[-, draw={rgb,255: red,0; green,162; blue,0}, dashed, line width=1pt]
\tikzstyle{blue filled}=[-, fill={blue!20}, draw=blue, line width=1pt, opacity=0.5, tikzit fill=white]
\tikzstyle{red filled}=[-, fill={red!20}, line width=1pt, draw=red, opacity=0.5, tikzit fill=white]
\tikzstyle{green filled}=[-, line width=1pt, draw={rgb,255: red,0; green,107; blue,61}, opacity=0.5, tikzit fill={rgb,255: red,149; green,255; blue,179}, fill={rgb,255: red,149; green,255; blue,179}]
\tikzstyle{orange filled}=[-, fill={orange!20}, draw=orange, line width=1pt, opacity=0.5, tikzit fill=white]
\tikzstyle{undirected dashed}=[-, draw=black, dashed, line width=1pt]
\tikzstyle{custom dotted}=[-, draw=black, dashed, line width=1pt, dash phase=3pt, shorten >=3pt]

\newtheorem{theorem}{Theorem}
\newtheorem{conjecture}{Conjecture}
\newtheorem{lemma}{Lemma}
\newtheorem{remark}{Remark}

\newtheorem{problem}{Problem}

\newcommand{\rot}{\intercal}                    
\newcommand{\transpose}{\intercal}                    
\newcommand{\inner}[2]{\left\langle #1 , #2 \right\rangle} 
\newcommand{\norm}[1]{\left\| #1\right\|}                  

\newcommand{\calL}{\mathcal{L}}

\newcommand{\calI}{\mathcal{I}}

\newcommand{\R}{\mathbb{R}}

\newcommand{\union}{\cup}
\newcommand{\intersect}{\cap}
\newcommand{\abs}[1]{\left\lvert#1\right\rvert}
\newcommand{\cardinality}[1]{\abs{#1}}

\newcommand{\reductionsatold}{\textsf{NAE-3SAT}}
\newcommand{\NP}{\mathsf{NP}}

\newcommand{\FNP}{\mathsf{FNP}}

\newcommand{\KSc}{\mathsf{KS}_2(c)}
\newcommand{\KScNP}{\mathsf{KS}_2\left(1/\left(4\sqrt{2}\right)\right)}

\newcommand{\partition}{\mathsf{PARTITION}}

\newcommand{\reductionsatnew}{\textsf{NAE-3SAT-KS}}

\newcommand{\sgn}{\text{sgn}}

\newcommand{\bigo}[1]{O\!\left(#1\right)}

\newcommand{\polylog}[1]{\mathrm{polylog}\!\left(#1\right)}

\definecolor{indiagreen}{rgb}{0.07, 0.53, 0.03}

\newcommand*{\horzbar}{\rule[.5ex]{2.5ex}{0.5pt}}

\newcommand{\twopartdefow}[3]
{
	\left\{
		\begin{array}{ll}
			#1 & \mbox{if } #2 \\
			#3 & \mbox{otherwise}
		\end{array}
	\right.
}

 \title{Is the Algorithmic Kadison-Singer Problem Hard?\footnote{This work is supported by an EPSRC Doctoral Training Studentship~(2590711),   and  an EPSRC  Fellowship~(EP/T00729X/1).}}

\author{anonymous}
\date{}

\author{Ben Jourdan\\ University of Edinburgh  \and Peter Macgregor \\ University of Edinburgh   \and He Sun \\ University of Edinburgh }

\begin{document}

\clearpage
 
\maketitle

\thispagestyle{empty}

\setcounter{page}{0}

\begin{abstract}

We study the following   $\KSc$ problem:
    let $c\in\mathbb{R}^+$ be some constant, and  $v_1,\ldots, v_m\in\mathbb{R}^d$  be   vectors 
such that $\|v_i\|^2\leq \alpha$ for any $i\in[m]$ and 
$\sum_{i=1}^m \langle v_i, x\rangle^2 =1$
for any $x\in\mathbb{R}^d$ with $\|x\|=1$.
The  $\mathsf{KS}_2(c)$ problem asks to    
  find  some   $S\subset [m]$,  such that it holds for   all $x \in \R^d$ with $\norm{x} = 1$  that
  \[
        \abs{\sum_{i \in S} \langle v_i, x\rangle^2 - \frac{1}{2}} \leq c\cdot\sqrt{\alpha},
        \]
or report no if such $S$ doesn't exist. Based on the   work of Marcus et al.~\cite{MSS} and Weaver~\cite{dm/Weaver04},   the $\KSc$ problem can be seen as the algorithmic   Kadison-Singer problem with   parameter $c\in\mathbb{R}^+$.

Our first  result  is 
  a  randomised algorithm with one-sided error for the    $\KSc$ problem such that (1) our algorithm finds a valid set $S\subset [m]$ with probability at least $1-2/d$, if such $S$ exists, or (2) reports no with probability $1$, if no valid sets exist. The algorithm has running time \[
\bigo{\binom{m}{n}\cdot \mathrm{poly}(m, d)}~\mbox{ for }~n = \exponentsmallc,
\]
 where $\epsilon$ is a parameter which controls the error of the algorithm. This presents the first algorithm for the Kadison-Singer problem whose running time is quasi-polynomial in $m$ in a certain regime, although having exponential dependency on $d$. 
Moreover, it shows that the algorithmic Kadison-Singer problem is easier to solve in low dimensions. 

Our second result is on  the computational complexity of the $\KSc$ problem. We show that 
 the
 $\KScNP$ problem is $\mathsf{FNP}$-hard for general values of $d$, and solving the $\KScNP$ problem is as hard as solving the  \reductionsatold\ problem.  
 \end{abstract}

\newpage

\section{Introduction\label{sec:intro}}
The Kadison-Singer problem~\cite{KS59}   posed in 1959 asks whether every pure state on the (abelian) von Neumann algebra $\mathbb{D}$ of bounded diagonal operators on $\ell_2$ has a unique extension to a pure state on $B(\ell_2)$, the von Neumann algebra of all bounded linear operators on the Hilbert space $\ell_2$. The statement of the Kadison-Singer problem 
arises from work on the foundations of quantum mechanics done by   Dirac in 1940s, and has been  subsequently shown to   be equivalent to numerous important problems in pure mathematics, applied mathematics, engineering and computer science~\cite{KS-detailed}.  
Weaver~\cite{dm/Weaver04} shows that
the Kadison-Singer problem  
 is equivalent to the following discrepancy question, which is    originally  posed as a conjecture.
 
\begin{conjecture}[The $\mathsf{KS}_2$ Conjecture]
There exist universal constants $\eta\geq 2$ and $\theta>0$ such that the following holds. Let $v_1,\ldots, v_m\in\mathbb{C}^d$ satisfy $\|v_i\|\leq 1$ for all $i\in[m]$, and suppose
$
\sum_{i=1}^m | \langle u, v_i \rangle |^2 = \eta
$
for every unit vector $u\in\mathbb{C}^d$. Then, there exists a partition $S_1, S_2$ of $[m]$ so that 
\[
\sum_{i\in S_{j}} |\langle u, v_i \rangle|^2 \leq \eta -\theta,\]
for every unit vector $u\in\mathbb{C}^d$ and every $j=\{1,2\}$.
\end{conjecture}

As a major breakthrough in mathematics, Marcus, Spielman and Srivastava~\cite{MSS} prove that the $\mathsf{KS}_2$ conjecture holds, and give an affirmative answer to the Kadison-Singer problem. Specifically, in this celebrated paper they show that, for any vectors $v_1,\ldots, v_m\in\mathbb{C}^d$  
  such that $\|v_i\|^2\leq\alpha$ for any $i\in [m]$ and $
\sum_{i=1}^m \langle v_i, x\rangle^2 =1$ 
for any $x\in\mathbb{C}^d$ with $\|x\|=1$, there is a partition $S_1, S_2$ of $[m]$ such that it holds for any $x\in\mathbb{C}^d$ with $\|x\|=1$ and 
$j=1,2$ that 
  \[
  \abs{\sum_{i \in S_j} \langle v_i, x\rangle^2 - \frac{1}{2}} \leq 3\cdot\sqrt{\alpha}.
  \]
The proof of this result is based on studying interlacing families of polynomials~\cite{MSS13}. While analysing interlacing families of polynomials suffices to answer the $\mathsf{KS}_2$ conjecture   and, as a consequence,
solve the Kadison-Singer problem, it is unclear if  their existential proof on the partition guaranteed by the $\mathsf{KS}_2$ conjecture can be turned into an efficient algorithmic construction; designing efficient algorithms for the Kadison-Singer problem is  listed as a natural open question in \cite{MSS}. 
This question is  particularly interesting in theoretical computer science, since it is directly linked  to constructing unweighted spectral sparsifiers~\cite{BSS} and   spectrally thin trees~\cite{AnariG14a}, among many other applications in approximation algorithms. However, there has been little work on the algorithmic Kadison-Singer problem,
and the complexity status of this problem is an important open question.

To address this question,  we study the following $\mathsf{KS}_2$ problem with some constant $c\in\mathbb{R}^+$:

\begin{problem}[The \textsf{KS}$_2(c)$ problem]\label{pro:main}
Given vectors $v_1,\ldots, v_m\in\mathbb{R}^d$ such that $\|v_i\|^2\leq \alpha$ for any $i\in[m]$ and 
$
\sum_{i=1}^m \langle v_i, x\rangle^2 =1$
for any $x\in\mathbb{R}^d$ with $\|x\|=1$, the {$\mathsf{KS}_2(c)$} problem asks to     
\begin{itemize}
    \item   find  some   $S\subset [m]$,  such that it holds for   all $x \in \R^d$ with $\norm{x} = 1$  that 
 \begin{equation}\label{eq:KS_condition}
        \abs{\sum_{i \in S} \langle v_i, x\rangle^2 - \frac{1}{2}} \leq c\cdot\sqrt{\alpha},
\end{equation}
    \item or report no if such $S$ doesn't exist.
\end{itemize}
\end{problem}

Notice that the $\mathsf{KS}_2$ conjecture
is equivalent to finding some subset $S\subset[m]$ as stated in 
Problem~\ref{pro:main} for some constant $c$. Here we choose to formulate the discrepancy of any set $S\subset [m]$ in \eqref{eq:KS_condition} as $c\cdot\sqrt{\alpha}$  for three   reasons: first of all,  Weaver~\cite{dm/Weaver04} shows  that 
the dependency on $O(\sqrt{\alpha})$
in \eqref{eq:KS_condition}
is tight, so the term  $O(\sqrt{\alpha})$ is unavoidable when bounding the discrepancy; secondly,    the $\mathsf{KS}_2$ conjecture shows that the existence of any universal constant $c$ in \eqref{eq:KS_condition} suffices to prove the     Kadison-Singer conjecture, and it is proven in  \cite{MSS} that the $\mathsf{KS}_2$ conjecture holds for $c=3$; however, 
studying the tightness of this constant remains an interesting open question on its own~(Problem~8.1, \cite{KSconsequence}).  Finally, as we  will show shortly,   the $\KSc$ problem belongs to different complexity classes with respect to different values of $c$, so introducing this parameter $c$ allows us to better understand the   complexity of the algorithmic Kadison-Singer problem.

\subsection{Our Results}
 
  Our first result is an algorithm called \algname$(\{v_i\}, c, \epsilon)$ for approximately solving the $\KSc$ problem for general values of $c$.   For any constant $c$,   $\epsilon<1$, and any vectors $v_1,\ldots,v_m\in\mathbb{R}^d$ such that $\|v_i\|^2\leq \alpha$ for all $i\in[m]$, we show that
  \begin{itemize}
      \item if there exists an $S$ which satisfies \eqref{eq:KS_condition}, then with probability at least $(1-2/d)$ the algorithm returns a set $S'\subset \{v_i\}_{i=1}^m$ that satisfies
  \begin{align}
      (1-\epsilon)\Big(\frac 1 2 - c\sqrt \alpha\Big)\le \sum_{v\in S'} \inner{v}{x}^2\le (1+\epsilon)\Big(\frac 1 2 + c\sqrt \alpha\Big) \label{eq:KS_ep_condition}
  \end{align}
  for all unit vectors $x \in \R^d$, and
      \item if no set exists which satisfies \eqref{eq:KS_ep_condition}, then with probability $1$ the algorithm returns `no'.
  \end{itemize}
  Our result is summarised as follows:
  \begin{theorem} \label{thm:algorithm}
  There is an algorithm,   \algname$(\mathcal{I}, c, \epsilon)$, such that for any instance  $\mathcal{I} \triangleq \{v_i\}_{i = 1}^m$  of the $\KSc$ problem with $v_i \in \R^d$ for $d\geq 3$,      and for any $\epsilon \in (0, 1)$, the following holds:
\begin{itemize}
\item
 if there exists a set $S \subset \mathcal{I}$ such that
\[
\left(\frac{1}{2} - c \sqrt{\alpha}\right) \leq  \sum_{v \in S} \inner{v}{x}^2  \leq \left(\frac{1}{2} + c \sqrt{\alpha} \right) 
\]
for all unit vectors $x \in \R^d$,
then with probability at least $(1 - 2/d)$, 
the \algname$(\mathcal{I}, c, \epsilon)$ algorithm returns a subset $S' \subset \mathcal{I}$ which satisfies~\eqref{eq:KS_ep_condition} for all unit vectors $x \in \R^d$.
\item if there is no set $S \subset \mathcal{I}$ which satisfies \eqref{eq:KS_ep_condition}, 
then with probability $1$, the \algname$(\mathcal{I}, c, \epsilon)$ algorithm reports that no such set exists.
\end{itemize}
The algorithm has running time 
\[
\bigo{\binom{m}{n}\cdot \mathrm{poly}(m, d)}~\mbox{ for }~n \triangleq \fullexponent.
\]
\end{theorem}

  \begin{remark}
 Since the most interesting instances of the 
 $\KSc$ problem are the cases in which $1/2 + c \sqrt{\alpha}$ is bounded away from $1$, we can assume that $c\sqrt{\alpha} \leq 1/2 - \sigma$ for some constant $\sigma$ which implies that
 \[
    n = \exponentsmallc.
 \]
Combining  this with  $d=\sum_{i=1}^m \|v_i\|^2\leq \alpha m$, a constraint due to the isotropic nature of the input, shows that our algorithm runs in quasi-polynomial time in $m$ when $d = \bigo{\polylog{m}}$.
 \end{remark}

Compared with 
the state-of-the-art  that runs in 
  $d^{O(m^{1/3} \alpha^{-1/4})}$ 
 time~\cite{conf/soda/AnariGSS18}, 
the most appealing fact of 
Theorem~\ref{thm:algorithm}
is that it shows the $\KSc$ problem can be approximately solved in quasi-polynomial time when $d=O(\mathrm{poly} \log m)$.  Moreover, for small values of $c$ where  
a subset $S\subset [m]$ satisfying \eqref{eq:KS_condition} isn't guaranteed to exist, our algorithm, with the same time complexity, is still able to find an $S$ satisfying \eqref{eq:KS_ep_condition} with high probability if it exists, or report no with probability $1$ otherwise.  These two facts together show that both determining the existence of a valid subset $S$   and finding such $S$ are computationally much easier  in low dimensions, regardless of the range of $c$.
In addition, our result is much stronger than a random sampling based algorithm, which only works in the regime of $\alpha=O(1/\log d)$~\cite{tropp2012user}, while our algorithm works even when there are vectors with much larger norm, e.g., $\alpha=\Theta(1)$.
On the other side, like   
  many   optimisation problems
  that involve the dimension of input items in their formulation~(e.g.,    multi-dimensional  packing~\cite{siamcomp/ChekuriK04}, and vector scheduling~\cite{algorithmica/BansalOVZ16}), 
    Theorem~\ref{thm:algorithm}  indicates that the order of $d$  might play a significant role in the hardness of the $\KSc$ problem, and the hard instances of the problem might be in the regime of $m=O(d)$.

Inspired by this,  we study the computational complexity of the $\KSc$ problem for general values of $d$,  where the number of input vectors satisfies $m=O(d)$.  
In order to study the `optimal' partitioning, for a given instance of the problem $\mathcal{I} = \{v_1, \ldots, v_m\}$, let
\[
    \mathcal{W}(\mathcal{I}) \triangleq \min_{S \subset \mathcal{I}} \max_{\substack{x \in \R^d\\ \|x\| = 1}} \abs{\sum_{v \in S} \langle v, x \rangle^2 - \frac{1}{2}}.
\]
Then, we choose $c=1/(4\sqrt{2})$ and notice that, for any vectors that satisfy the conditions of the $\KSc$ problem, there could be no subset $S$ satisfying \eqref{eq:KS_condition} for such $c$. As our second result, we prove that, for any  $c \leq 1/(4\sqrt{2})$,
distinguishing between instances for which $\mathcal{W}(\mathcal{I}) = 0$ and those for which $\mathcal{W}(\mathcal{I}) \geq c \cdot \sqrt{\alpha}$ is $\NP$-hard.  Our result is as follows:
\begin{theorem} \label{thm:hardness}
The $\mathsf{KS}_2\left(1/\left(4\sqrt{2}\right)\right)$ problem is   $\mathsf{FNP}$-hard  for general values of $d$.
Moreover, it is $\NP$-hard to distinguish between instances of the $\KSc$ problem with $\mathcal{W}(\mathcal{I}) = 0$ from instances with $\mathcal{W}(\mathcal{I}) \geq \left( 1 / 4 \sqrt{2} \right) \cdot \sqrt{\alpha}$.
\end{theorem}
\begin{remark}
It is important to note that, when $d$ is constant, the decision problem in Theorem~\ref{thm:hardness} can be solved in polynomial time.
For example, the $1$-dimensional problem is equivalent to the
$\partition$ problem, in which we are given a set of real numbers $\mathcal{I} = \{x_1, \ldots, x_m\}$ such that $\sum_i x_i = 1$ and must determine whether there exists a subset $S \subset \mathcal{I}$ such that $\sum_{x \in S} x = 1/2$.
In this setting, \[
\mathcal{W}(\mathcal{I}) = \min_{S \subset \mathcal{I}} \abs{\left(\sum_{x \in S} x\right) - 1/2}.
\]
There is a well-known FPTAS for $\partition$ which can distinguish between instances for which $\mathcal{W}(\mathcal{I}) = 0$ and those for which $\mathcal{W}(\mathcal{I}) \geq \epsilon$, for any $\epsilon > 0$.
An important consequence of Theorem~\ref{thm:hardness} is that there is no such FPTAS for the optimisation version of the $\KSc$ problem for general $d$.
\end{remark}

 Theorem~\ref{thm:hardness} shows that the
isotropic structure of the $\KSc$ instance is not sufficient to make finding a partition easy when compared with similar problems.
As such, the 
design of a potential polynomial-time algorithm for the Kadison-Singer problem would need to take some range of $c$ into account
and cannot solve the optimisation version of the $\KSc$ problem,
otherwise one would end up solving an $\mathsf{NP}$-hard problem.   We remark that Theorem~\ref{thm:hardness} shares the same style as the one for  Spencer’s Discrepancy Problem: given any input on $N$ elements, Charikar et al.~\cite{CNN11} shows that it is $\NP$-hard to distinguish between the input with discrepancy zero and the one with discrepancy $\Omega(\sqrt{N})$, although it is known that a solution with  $O(\sqrt{N})$ approximation can be computed efficiently~\cite{Bansal10}.



\subsection{Our Techniques}

In this subsection we sketch our main techniques used in proving Theorems~\ref{thm:algorithm} and \ref{thm:hardness}.
 \paragraph{Proof Sketch of Theorem~\ref{thm:algorithm}.} 

  We start by sketching the ideas
behind our algorithmic result. First of all, it is easy to see that
we can solve the $\KSc$ problem for any $c\in\mathbb{R}^+$ in $O\left(2^m\cdot\mathrm{poly}(m,d)\right)$ time, since we only need to enumerate all the $2^m$ subsets $S\subseteq \calI$ of the input set $\calI$ and check if every possible set $S$ satisfies the condition \eqref{eq:KS_condition}.  To express all the subsets of $\calI$, we 
inductively construct level sets $\left\{\calL_i\right\}_{i=0}^m$ with $\calL_i\subseteq 2^{\calI}$ as follows:

 \begin{itemize}
 \item initially, level $i=0$ consists of a single set $\emptyset$, and we set $\calL_0=\{\emptyset\}$;
 \item based on $\calL_{i-1}$ for any $1\leq i\leq m$, we define $\calL_{i}$ by
$
 \calL_{i} \triangleq  \left\{ S, S\cup\{v_{i}\} : S\in \calL_{i-1} \right\}$.
\end{itemize}
It is   important to see that, although $|\calL_i|$ could be as high as $2^m$, there are only $m$ such
level sets
$\calL_i$, which
are constructed inductively in an \emph{online} manner, and  it holds for any $S\subseteq \calI$  that $S\in \mathcal{L}_m$.

The bottleneck for
improving the efficiency of this simple enumeration algorithm 
is the number of sets in $\calL_m$, which could be exponential in $m$. To overcome this bottleneck,  we introduce the notion of \emph{spectral equivalence classes} to reduce $|\calL_i|$ for any $i\in[m]$. Informally speaking,   if there are different $S_1, S_2\in\calL_i$ for any $i\in[m]$ such that\footnote{For any two matrices $A$ and $B$ of the same dimension, we write $A\preceq B$ if   $B-A$ is positive semi-definite.}
\[
    (1 - \epsilon) \sum_{j \in S_2} v_j v_j^\transpose \preceq \sum_{j \in S_1} v_j v_j^\transpose \preceq (1 + \epsilon) \sum_{j \in S_2} v_j v_j^\transpose
\]
  for some small  $\epsilon$, then  we view  $S_1$ and $S_2$ to be  ``spectrally equivalent'' to each other\footnote{Although this relationship is not symmetric, this informal definition is sufficient for the proof sketch and is not used directly in our analysis.}. It suffices to use one   set to represent all of its spectral equivalences; hence, we only need to store the subsets which aren't spectrally equivalent to each other\footnote{The list of stored subsets can be thought of as an epsilon cover of all possible subsets.}.
  Since there is a spectral sparsifier of any $S$ with   $O(d\log (d)/\epsilon^2)$ vectors~\cite{cohen2016online,SpielmanS11}, we can   reduce the total number of stored subsets~(i.e., the
  number of    spectral equivalence classes) in $\calL_i$  for any $i\in[m]$   to $\binom{m}{n}$ where $n = \bigo{d \log(d) / \epsilon^2}$ which is no longer exponential in $m$. 

Turning this idea into an algorithm design, we need be careful that the small approximation error introduced by every constructed spectral sparsifier  does not compound as we construct sparsifiers from one level to another. 
In order to avoid this, we employ the online vector sparsification algorithm presented in  \cite{cohen2016online}. This allows us to construct   sparsifiers  in $\calL_{i}$ from the ones in $\calL_{i-1}$ and the vector $v_{i}$. In addition,  the construction in each level preserves the same approximation error as the previous one.

We highlight that the design of  our   algorithm for solving the $\KSc$ problem  is entirely different from the previous work, which is  based on analysing the properties of interlacing polynomials~\cite{conf/soda/AnariGSS18}.
Moreover, one can view our use of online spectral sparsifiers in constructing spectral equivalence classes as an \emph{encoding} strategy to reduce the enumeration space of the $\KSc$ problem. From this aspect, our work sheds light on potential  applications of other tools well-studied in algorithmic spectral graph theory and numerical linear algebra, such as sparsification and sketching.

\paragraph{Proof Sketch of Theorem~\ref{thm:hardness}.} 
Our proof of the $\mathsf{FNP}$-hardness of   the $\mathsf{KS}_2\left(1/\left(4\sqrt{2}\right)\right)$ problem is  based on a reduction from the well-known \reductionsatold\ problem~\cite{garey_computers_1979} to
a
decision version of the $\mathsf{KS}_2\left(1/\left(4\sqrt{2}\right)\right)$ problem,
which asks whether $\mathcal{W}(\mathcal{I}) = 0$ or $\mathcal{W}(\mathcal{I}) \geq \left(1 / \left(4 \sqrt{2}\right)\right) \sqrt{\alpha}$.
Our overall reduction consists of two steps: we first build a reduction from the \reductionsatold\ problem to the so-called \reductionsatnew\ problem, and then build a reduction from the \reductionsatnew\ problem to   the $\mathsf{KS}_2\left(1/\left(4\sqrt{2}\right)\right)$ problem.
 
 
To sketch the first reduction, we examine the so-called \reductionsatnew\ problem, which can be viewed as a restricted version of the \reductionsatold\ problem, and used only as a tool to build the reduction from the \reductionsatold\ problem to the $\KScNP$ problem. 
Informally, the \reductionsatnew\ problem consists of the \textsf{3SAT} Boolean formulae $\psi$, in which   the number of occurrences of both $u$ and $\bar{u}$ for every variable $u$ in any $\psi$ is limited with respect to some additional constraints and any two clauses of $\psi$ share  at most one literal; the \reductionsatnew\ problem asks if there is a satisfying assignment for $\psi$ such that every clause of $\psi$ has at least one true literal and at least one false literal; we refer the reader to  Problem~\ref{prob:satnew} in Section~\ref{sec:hardness} for the formal definition of the \reductionsatnew\ problem. Based on a reduction from the \reductionsatold\ problem, we show that the \reductionsatnew\ problem is $\NP$-complete.

For the second and main reduction of our analysis, we  build a   reduction from the \reductionsatnew\ problem to the $\mathsf{KS}_2\left(1/\left(4\sqrt{2}\right)\right)$ problem. Specifically, for an \reductionsatnew\ instance $\psi$ of $n$ variables and $m$ clauses, we construct a set $A$ of $\Theta(n+m)$ vectors as 
a $\mathsf{KS}_2\left(1/\left(4\sqrt{2}\right)\right)$ instance, and each   
$v\in A$ has dimension $n+m$, such that the following properties hold:
\begin{itemize}
    \item every vector $v$ has norm $\|v\|^2 \leq 1/4$ and \[\sum_{v\in A} v v^\transpose =I;\]
    \item if $\psi$ is a satisfiable instance of \reductionsatnew, then there is a subset $S\subset A$  such that \[\sum_{v\in S} v v^\transpose = (1/2)\cdot I;\]
    \item if $\psi$ is not a satisfiable instance of \reductionsatnew, then for any subset $S\subset A$ there is always some $y\in\mathbb{R}^n$ with $\|y\|=1$ such that 
    \[
    \left| 
    \sum_{v\in S} \inner{v}{y}^2 - \frac{1}{2}
    \right| \geq \frac{1}{8 \sqrt{2}}.
    \]
    \end{itemize}
The key to proving these properties is a careful construction of a $\mathsf{KS}$ instance $\mathcal{I}$ from any formula $\psi$, and an analysis of the properties of $\sum_{v\in S}v v^\transpose $ for any set $S \subseteq \mathcal{I}$ if $\psi$ is an unsatisfiable instance of \reductionsatnew.
We think that   such a reduction from any \textsf{SAT} instance to a \textsf{KS} instance is quite  novel, and might be further employed to sharpen the   constant $1/(4 \sqrt{2})$.

\subsection{Related Work}

There has been little   work on the algorithmic Kadison-Singer problem.
 Anari et al.~\cite{conf/soda/AnariGSS18} studies approximating  the largest root of a real rooted polynomial and its applications to interlacing families, which are the main tool developed in \cite{MSS} to prove the Kadison-Singer conjecture. They show that a valid partition promised by Weaver's $\mathsf{KS}_2$ conjecture can be found in  $d^{O(m^{1/3} \alpha^{-1/4})}$  time, suggesting that   exhaustive  
search of all possibilities is not required for the algorithmic Kadison-Singer problem  unlike 
the strong exponential time hypothesis for the \textsf{SAT} problem. 

 Becchetti et al.~\cite{conf/soda/BecchettiCNPT20} studies the algorithmic Kadison-Singer problem  for graphs under some restricted condition. Specifically, they show that, if $G=(V,E)$ is an $n$-vertex and $\Delta$-regular graph of
$\Delta=\Omega(n)$ and the second eigenvalue of the adjacency matrix of $G$ is at most a sufficient small constant times $\Delta$, then an unweighted spectral sparsifier of $G$ can be constructed efficiently.  Their algorithm is combinatorial and only works for graphs. 

Weaver~\cite{weaver2} shows that the   \textsf{BSS}-framework for constructing linear-sized spectral sparsifiers~\cite{BSS} can be adapted for the one-sided Kadison-Singer problem, where the  term ``one-sided'' refers to the fact that   the    discrepancy of the algorithm's output   can be only upper bounded.
   
Finally, independent of our work, Spielman and Zhang~\cite{SZ22} studies the same complexity problem as ours. Different from our approach, their analysis starts with the $(3,2\mbox{-}2)$ Set Splitting problem, which is a variant of  the $2$-$2$ Set Splitting problem. They prove that the $(3,2\mbox{-}2)$ Set Splitting problem remains \textsf{NP}-hard even if no pair of sets intersects   in more than one variable. Applying this, they show that the $\mathsf{KS}_2(c)$ problem is $\mathsf{NP}$-hard for  $c=1/4$. While their result is slightly tighter than ours with respect to the value of $c$, the conclusions of the two works are essentially the same. 


\subsection{Notation \label{sec:pre}}
 Let $[m]\triangleq \{1,\ldots, m\}$. For any integer $j$, we define vector $\mathbf{1}_j$, in which $\mathbf{1}_j(j)=1$ and all of $\mathbf{1}_j$'s other entries are $0$.  
 For any integer $d\geq 1$, let $\mathbf{0}_{d \times d}\in\mathbb{R}^{d\times d}$ be the  matrix in which every entry is equal to $0$.  
 
 We call a matrix $A$   positive semi-definite~(\textsf{PSD}) if  $x^{\rot}Ax\geq 0$ holds for any $x\in\mathbb{R}^d$. For any two matrices $A$ and $B$, we write $A\preceq B$ if $B-A$ is \textsf{PSD}. The  spectral norm of any matrix $A$ is expressed by $\|A\|$.




\section{Algorithm Based on Spectral Equivalence Classes \label{sec:algo}}
This section discusses in detail  the construction of spectral equivalence classes, and its application in designing   a randomised algorithm for the $\KSc$ problem.  We analyse the presented algorithm, and   prove Theorem~\ref{thm:algorithm}.  
\subsection{Algorithm}
Our algorithm consists of $m$ iterations:
in iteration $i$, the algorithm constructs the set $\calL_i$ of spectral equivalence classes for the subsets $S \subseteq \{v_1, \ldots, v_i\}$.
For each equivalence class, $\calL_i$ contains a pair $(S, B)$ where $S \subseteq \{v_1, \ldots v_i\}$ is a representative set in the equivalence class and $B \in \R^{d \times d}$ is a spectral sparsifier representing the equivalence class.
Moreover, the algorithm 
constructs the representations of spectral equivalence classes in iteration $i$ based on the ones maintained in iteration $i-1$. That is, instead of constructing all the subsets of $\{v_1, \ldots, v_{i}\}$ and grouping them into different spectral equivalence classes,
the algorithm directly constructs the representations of the spectral equivalence classes of $\{v_1, \ldots, v_{i}\}$ based on its constructed equivalence classes of $\{v_1, \ldots, v_{i-1}\}$.
This can be achieved by applying an online algorithm for constructing spectral sparsifiers, since,
if we assume that in iteration $i-1$ every subset $S \subseteq \{v_1, \ldots, v_{i-1}\}$ is spectrally equivalent to some $(S', B') \in \calL_{i-1}$ maintained by the algorithm, then both of $S$ and $S \cup \{v_{i}\}$ are spectrally equivalent to $S'$ and $S' \cup \{v_{i}\}$ in iteration $i$ as well. 
As such, in iteration $i$ we only need to ensure that the sets $S'$ and $S' \union \{v_{i}\}$ are still represented by some sparsifiers in $\calL_{i}$.

Based on this, we can view all the vectors $v_1,\ldots, v_m$ as arriving \emph{online} and, starting with the trivial spectral equivalence class defined by $\calL_0 = \{(\emptyset, \mathbf{0}_{d \times d})\}$, the algorithm constructs the representations of spectral equivalence classes of $\{v_1, \ldots v_i\}$ in iteration $i$.
Our algorithm applies the online algorithm for constructing spectral sparsifiers~\cite{cohen2016online}~(Lines~\ref{algline:lcomp}-\ref{algline:lupdate2} of Algorithm~\ref{alg:sparsified-ks})
to construct the representations of spectral equivalence classes of $\{v_1, \ldots, v_{i}\}$ based on those of $\{v_1, \ldots, v_{i-1}\}$. Since any subset of $\{v_1, \ldots v_m\}$ is spectrally equivalent to some set of  vectors with size
$n$ where $n$ is nearly linear in $d$~\cite{cohen2016online},
the number of spectral equivalence classes in any set $\calL_i$ will be at most $\binom{m}{n}$.
See Figure~\ref{fig:sparsified_table} for an illustration of the construction of the sets $\calL_i$ and Algorithm~\ref{alg:sparsified-ks} for the formal description of the algorithm.


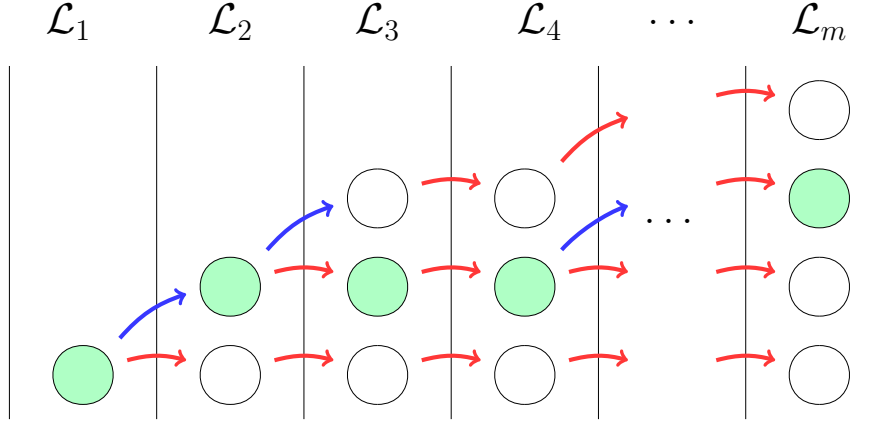
\begin{figure}[ht]
    \centering
    \resizebox{12cm}{!}{%
    \begin{tikzpicture}
	\begin{pgfonlayer}{nodelayer}
		\node [style=none] (0) at (5, 0) {};
		\node [style=none] (1) at (5, 12) {};
		\node [style=none] (2) at (10, 0) {};
		\node [style=none] (3) at (10, 12) {};
		\node [style=none] (4) at (15, 0) {};
		\node [style=none] (5) at (15, 12) {};
		\node [style=none] (6) at (20, 0) {};
		\node [style=none] (7) at (20, 12) {};
		\node [style=none] (20) at (7, 13.5) {\huge$\mathcal{L}_1$};
		\node [style=none] (30) at (9, 2) {};
		\node [style=none] (31) at (11, 2) {};
		\node [style=none] (32) at (11, 4.25) {};
		\node [style=none] (33) at (8.75, 2.75) {};
		\node [style=none] (42) at (7.5, 0.5) {};
		\node [style=none] (43) at (7.5, 2.5) {};
		\node [style=none] (44) at (12.5, 13.5) {\huge$\mathcal{L}_2$};
		\node [style=none] (45) at (12.5, 0.5) {};
		\node [style=none] (46) at (12.5, 2.5) {};
		\node [style=none] (47) at (12.5, 3.5) {};
		\node [style=none] (48) at (12.5, 5.5) {};
		\node [style=none] (49) at (17.5, 0.5) {};
		\node [style=none] (50) at (17.5, 2.5) {};
		\node [style=none] (51) at (17.5, 3.5) {};
		\node [style=none] (52) at (17.5, 5.5) {};
		\node [style=none] (53) at (17.5, 6.5) {};
		\node [style=none] (54) at (17.5, 8.5) {};
		\node [style=none] (55) at (14, 2) {};
		\node [style=none] (56) at (16, 2) {};
		\node [style=none] (57) at (14, 5) {};
		\node [style=none] (58) at (16, 5) {};
		\node [style=none] (59) at (16, 7.25) {};
		\node [style=none] (60) at (13.75, 5.75) {};
		\node [style=none] (61) at (25, 0) {};
		\node [style=none] (62) at (25, 12) {};
		\node [style=none] (63) at (22.5, 0.5) {};
		\node [style=none] (64) at (22.5, 2.5) {};
		\node [style=none] (65) at (22.5, 3.5) {};
		\node [style=none] (66) at (22.5, 5.5) {};
		\node [style=none] (67) at (22.5, 6.5) {};
		\node [style=none] (68) at (22.5, 8.5) {};
		\node [style=none] (69) at (19, 2) {};
		\node [style=none] (70) at (21, 2) {};
		\node [style=none] (71) at (19, 5) {};
		\node [style=none] (72) at (21, 5) {};
		\node [style=none] (73) at (19, 8) {};
		\node [style=none] (74) at (21, 8) {};
		\node [style=none] (75) at (17.5, 13.5) {\huge$\mathcal{L}_3$};
		\node [style=none] (76) at (23, 13.5) {\huge$\mathcal{L}_4$};
		\node [style=none] (79) at (32.5, 0.5) {};
		\node [style=none] (80) at (32.5, 2.5) {};
		\node [style=none] (81) at (32.5, 3.5) {};
		\node [style=none] (82) at (32.5, 5.5) {};
		\node [style=none] (83) at (32.5, 9.5) {};
		\node [style=none] (84) at (32.5, 11.5) {};
		\node [style=none] (85) at (24, 2) {};
		\node [style=none] (86) at (26, 2) {};
		\node [style=none] (87) at (24, 5) {};
		\node [style=none] (88) at (26, 5) {};
		\node [style=none] (89) at (23.75, 8.75) {};
		\node [style=none] (90) at (26, 10.25) {};
		\node [style=none] (92) at (32.5, 6.5) {};
		\node [style=none] (93) at (32.5, 8.5) {};
		\node [style=none] (94) at (26, 7.25) {};
		\node [style=none] (95) at (23.75, 5.75) {};
		\node [style=none] (107) at (27.5, 6.75) {\huge $\dots$};
		\node [style=none] (108) at (30, 0) {};
		\node [style=none] (109) at (30, 12) {};
		\node [style=none] (110) at (32.5, 13.5) {\huge$\mathcal{L}_m$};
		\node [style=none] (119) at (29, 2) {};
		\node [style=none] (120) at (31, 2) {};
		\node [style=none] (121) at (29, 5) {};
		\node [style=none] (122) at (31, 5) {};
		\node [style=none] (123) at (31, 8) {};
		\node [style=none] (124) at (29, 8) {};
		\node [style=none] (127) at (31, 11) {};
		\node [style=none] (128) at (29, 11) {};
		\node [style=none] (142) at (27.5, 13.5) {\huge$\ldots$};
		\node [style=none] (143) at (35, 0) {};
		\node [style=none] (144) at (35, 12) {};
	\end{pgfonlayer}
	\begin{pgfonlayer}{edgelayer}
		\draw (1.center) to (0.center);
		\draw (3.center) to (2.center);
		\draw (5.center) to (4.center);
		\draw (7.center) to (6.center);
		\draw [style=directed red, bend left=15] (30.center) to (31.center);
		\draw [style=directed blue, bend left=15] (33.center) to (32.center);
		\draw [style=filled set cell] (42.center)
			 to [bend left=90, looseness=1.75] (43.center)
			 to [bend left=90, looseness=1.75] cycle;
		\draw (45.center)
			 to [bend left=90, looseness=1.75] (46.center)
			 to [bend left=90, looseness=1.75] cycle;
		\draw [style=filled set cell] (47.center)
			 to [bend left=90, looseness=1.75] (48.center)
			 to [bend left=90, looseness=1.75] cycle;
		\draw [bend left=90, looseness=1.75] (49.center) to (50.center);
		\draw [bend left=90, looseness=1.75] (50.center) to (49.center);
		\draw [style=filled set cell] (52.center)
			 to [bend left=90, looseness=1.75] (51.center)
			 to [bend left=90, looseness=1.75] cycle;
		\draw [bend left=90, looseness=1.75] (53.center) to (54.center);
		\draw [bend left=90, looseness=1.75] (54.center) to (53.center);
		\draw [style=directed red, bend left=15] (55.center) to (56.center);
		\draw [style=directed red, bend left=15] (57.center) to (58.center);
		\draw [style=directed blue, bend left=15] (60.center) to (59.center);
		\draw (62.center) to (61.center);
		\draw [bend left=90, looseness=1.75] (63.center) to (64.center);
		\draw [bend left=90, looseness=1.75] (64.center) to (63.center);
		\draw [style=filled set cell] (66.center)
			 to [bend left=90, looseness=1.75] (65.center)
			 to [bend left=90, looseness=1.75] cycle;
		\draw [bend left=90, looseness=1.75] (67.center) to (68.center);
		\draw [bend left=90, looseness=1.75] (68.center) to (67.center);
		\draw [style=directed red, bend left=15] (69.center) to (70.center);
		\draw [style=directed red, bend left=15] (71.center) to (72.center);
		\draw [style=directed red, bend left=15] (73.center) to (74.center);
		\draw [bend left=90, looseness=1.75] (79.center) to (80.center);
		\draw [bend left=90, looseness=1.75] (80.center) to (79.center);
		\draw [bend left=90, looseness=1.75] (81.center) to (82.center);
		\draw [bend left=90, looseness=1.75] (82.center) to (81.center);
		\draw [bend left=90, looseness=1.75] (83.center) to (84.center);
		\draw [bend left=90, looseness=1.75] (84.center) to (83.center);
		\draw [style=directed red, bend left=15] (85.center) to (86.center);
		\draw [style=directed red, bend left=15] (87.center) to (88.center);
		\draw [style=directed red, bend left=15] (89.center) to (90.center);
		\draw [style=filled set cell] (92.center)
			 to [bend left=90, looseness=1.75] (93.center)
			 to [bend left=90, looseness=1.75] cycle;
		\draw [style=directed blue, bend left=15, looseness=0.75] (95.center) to (94.center);
		\draw (109.center) to (108.center);
		\draw [style=directed red, bend left=15] (119.center) to (120.center);
		\draw [style=directed red, bend left=15] (121.center) to (122.center);
		\draw [style=directed red, bend left=15] (124.center) to (123.center);
		\draw [style=directed red, bend left=15] (128.center) to (127.center);
		\draw (144.center) to (143.center);
	\end{pgfonlayer}
\end{tikzpicture}
    }
    \caption{The construction of the sets $\mathcal{L}_i$ in Algorithm~\ref{alg:sparsified-ks}.
    Each $\calL_{i-1}$ contains sparsifiers representing the spectral equivalence classes of the vectors $\{v_1, \ldots, v_{i-1}\}$.
    Then, $\calL_{i}$ contains either one or two ``children'' of each sparsifier in $\calL_{i-1}$, where the second child is added with some small probability which prevents $\cardinality{\calL_m}$ from growing exponentially with $m$.
    For a particular target subset $S \subseteq \{v_1, \ldots v_m\}$, there is some sequence of constructed sparsifiers which corresponds to the process of the online algorithm for constructing spectral sparsifiers~\cite{cohen2016online}, applied to $S$.}
    \label{fig:sparsified_table}
\end{figure}

\begin{algorithm}[H] \label{alg:sparsified-ks}
\caption{\algname$(\mathcal{I} = \{v_i\}_{i = 1}^m, c, \epsilon)$, where $v_i \in \R^d$ and $\norm{v_i}^2 \leq \alpha$ }
$\mu \gets \epsilon / 6$ \\
$\lambda \gets \min\left(c \sqrt{\alpha}, 1/2 - c \sqrt{\alpha}\right)$ \\
$b \gets 8 \log(d) / \mu^2$ \\
 $n \gets \bigo{d \log(d) \log(1/\lambda) / \mu^2}$\\

$\calL_0 \gets \{ (\emptyset, \mathbf{0}_{d \times d}) \}$ \\ \label{algline:basecase}

\For{$i \gets 1~\mathrm{\ to\ }~m$}{
    $\calL_i \gets \emptyset$ \\
    \For{$(S, B) \in \calL_{i-1}$ and $B$ is constructed with at most $n$ vectors \label{algline:loop}}{
                $S' \gets S \union \{v_i\}$ \\
                \If{$S'$ satisfies \eqref{eq:KS_ep_condition} \label{algline:if}}{
                        \Return $S'$
                }
                $p \gets \min\left(b \left(1 + \mu\right) v_i^\transpose \left(B + \lambda I \right)^{-1} v_i, 1 \right)$ \\ \label{algline:lcomp}
                \eIf{$X \leq p$ where $X \sim \mathrm{Uniform}[0, 1]$ \label{algline:randif}}{
                    $B' \gets B + \frac{1}{p}v_i v_i^\transpose$ \\
                    $\calL_i \gets \calL_i \union \{(S, B), (S', B')\}$ \\ \label{algline:lupdate}
                }{
                    $\calL_i \gets \calL_i \union \{(S', B)\}$ \\ \label{algline:lupdate2}
                }
    }
}
\Return \textsc{Failure}
\end{algorithm}

\begin{remark}
 The if-condition on Line~\ref{algline:if} of Algorithm~\ref{alg:sparsified-ks} can be checked in polynomial time  while introducing an arbitrarily small error, by constructing the matrix $\sum_{v\in S'} vv^{\rot}$ and computing its eigenvalues. 
\end{remark}

\subsection{Analysis}
Since it holds for any   vectors $v_1,\ldots, v_{\ell}$ with $v_i\in\mathbb{R}^d$ that
\[
\sum_{i=1}^{\ell} v_iv_i^{\rot} = 
\begin{pmatrix}
\horzbar\ v_1\ \horzbar\\ 
\horzbar\ v_2\ \horzbar\\ 
\hspace{0.25cm}\vdots\hspace{0.25cm}\\
\horzbar\ v_{\ell}\ \horzbar 
\end{pmatrix}^{\rot}
\begin{pmatrix}
\horzbar\ v_1\ \horzbar\\ 
\horzbar\ v_2\ \horzbar\\ 
\hspace{0.25cm}\vdots\hspace{0.25cm}\\
\horzbar\ v_{\ell}\ \horzbar 
\end{pmatrix},
\]
sparsifying $\sum_{v\in S} v v^{\rot}$ for any $S \subseteq \mathcal{I}$ is equivalent to sparsifying the $|S|\times d$ matrix whose rows are defined by all the $v\in S$. Based on this, our proof uses the result from the online matrix sparsification algorithm~\cite{cohen2016online} as a black box. Specifically, we apply the following lemma in our analysis, which is a special case of Theorem~2.3 from~\cite{cohen2016online}. Notice that the algorithm described in Lemma~\ref{lem:cohenonline} below corresponds to the sampling scheme used in Algorithm~\ref{alg:sparsified-ks}.

\begin{lemma}[\cite{cohen2016online}, Theorem 2.3] \label{lem:cohenonline}
Let $S$ be a set of vectors $v_1, \ldots, v_m \in \R^d$, and let $A = \sum_{v \in S} v v^\transpose$. With $\mu, \delta \in [0, 1]$,  $b \triangleq 8 \log(d) / \mu^2$ and $B_0 = \mathbf{0}_{d \times d}$, construct $B_i$ inductively for  $i \in [m]$ such that 
with probability
    \[
        p_i = \min\left(b (1 + \mu) v_i^\transpose \left(B_{i - 1} + \frac{\delta}{\mu} I\right)^{-1} v_i, 1\right),
    \]
we have
    \[
        B_i = B_{i - 1} + \frac{1}{p_i} v_i v_i^\transpose,
    \]
    and with probability $1 - p_i$, we have $B_i = B_{i - 1}$.
    Then, it holds with  probability $(1 - 1/d)$ that
    \[
        (1 - \mu) A - \delta I \preceq B_m \preceq (1 + \mu) A + \delta I,
    \]
    and the number of vectors added to $B_m$ is $O\left(d\log d \log\left(\mu \|A\|^2/\delta \right)/\mu^2 \right )$.
\end{lemma}

Now, we analyse Algorithm~\ref{alg:sparsified-ks}.
We begin by showing that for each pair $(S, B)$ constructed by Algorithm~\ref{alg:sparsified-ks}, $B$ is a spectral sparsifier of $S$ with high probability.
\begin{lemma} \label{lem:algpairs}
Let $\mathcal{L}_i$ be the set constructed by Algorithm~\ref{alg:sparsified-ks} at iteration $i$.
Then, for any $(S^\star, B^\star) \in \mathcal{L}_i$, it holds with probability $(1 - 1/d)$ that
\[
(1 - \mu) A_{S^\star} - \delta I \preceq B^\star \preceq (1 + \mu) A_{S^\star} + \delta I
\]
where $A_{S^\star} = \sum_{v \in S^\star} v v^\transpose$, and the parameters are set in Algorithm~\ref{alg:sparsified-ks} to be $\mu = \epsilon / 6$ and \[\delta = \mu \min(c \sqrt{\alpha}, 1/2 - c \sqrt{\alpha}).\]
\end{lemma}

\begin{proof}
We will show that for any pair $(S^\star, B^\star)$ constructed by Algorithm~\ref{alg:sparsified-ks}, $B^\star$ is equivalent to the output of the algorithm described in Lemma~\ref{lem:cohenonline} when applied to $S^\star$.
We prove this by induction on $i$.
The base case $i = 0$ follows immediately from the initialisation of $\mathcal{L}_0 = \{(\emptyset, \mathbf{0}_{d \times d})\}$.
For the inductive step we show that the conclusion holds for every pair in $\mathcal{L}_{i}$, assuming it holds for every pair in $\mathcal{L}_{i-1}$.
For each pair $(S^\star, B^\star) \in \mathcal{L}_{i}$,
the proof proceeds by a case distinction.

\textbf{Case 1: $(S^\star, B^\star) \in \mathcal{L}_{i-1}$.}
This case corresponds to the pairs $(S, B)$ added on Line~\ref{algline:lupdate} of Algorithm~\ref{alg:sparsified-ks}.
Accordingly, by the inductive hypothesis, we have that $B^\star$ is equivalent to the output of the algorithm described in Lemma~\ref{lem:cohenonline} applied to $S^\star$.

\textbf{Case 2: $(S^\star, B^\star) \not \in \mathcal{L}_{i-1}$.}
This case covers the pairs involving $S'$ added on Lines~\ref{algline:lupdate} and \ref{algline:lupdate2} of Algorithm~\ref{alg:sparsified-ks}.
Let $(S, B_{i-1})$ be the pair in $\mathcal{L}_{i-1}$ from which $(S^\star, B^\star)$ is constructed.
Notice that $S^\star = S \union \{v_{i}\}$.
Then, by the construction of Algorithm~\ref{alg:sparsified-ks}, with probability $p_{i}$, we have
\[
    B^\star = B_{i-1} + \frac{1}{p_{i}} v_{i}{v_{i}}^\transpose
\]
and with probability $1 - p_{i}$, we have $B^\star = B_{i-1}$, where $p_{i}$ is the probability defined in Lemma~\ref{lem:cohenonline}.
As such, $B^\star$ is the result of applying an iteration of the algorithm defined in Lemma~\ref{lem:cohenonline}, for the new vector $v_{i}$.
This maintains that $B^\star$ is equivalent to the output of the Lemma~\ref{lem:cohenonline} algorithm applied to $S^\star$ and completes the inductive argument.
\end{proof}

We now show that \emph{any} set $S \subset \{v_1, \ldots, v_m\}$ is well approximated by one of the sparsifiers constructed in Algorithm~\ref{alg:sparsified-ks}.

\begin{lemma} \label{lem:ssparsifier}
Let $\mathcal{I} = \{v_i\}_{i = 1}^m$ be the input to Algorithm~\ref{alg:sparsified-ks}.     Let $S \subseteq \mathcal{I}$ be any fixed set, and $A = \sum_{v \in S} v v^\transpose$.
Then, with
probability $(1 - 1/d)$,
there is a matrix $B$ constructed by Algorithm~\ref{alg:sparsified-ks} such that
\[
    (1 - \mu) A - \delta I \preceq B \preceq (1 + \mu) A + \delta I,
\]
where $\mu = \epsilon / 6$ and $\delta = \mu \min(c \sqrt{\alpha}, 1/2 - c \sqrt{\alpha})$.
\end{lemma}

\begin{proof}
We prove that one of the matrices $B$ constructed by Algorithm~\ref{alg:sparsified-ks} is equivalent to the output of the algorithm defined in Lemma~\ref{lem:cohenonline} applied to the set $S$.
Although the matrices constructed in Algorithm~\ref{alg:sparsified-ks} are always part of a pair $(S', B)$, in this proof we consider only the matrices $B$, and ignore the sets $S'$ which are constructed alongside them.

We now inductively define a sequence $B_0, B_1, \ldots, B_{m}$, such that $B_i$ is a matrix constructed by the algorithm in iteration $i$ and $B_i \in \mathcal{L}_i$ corresponds to the output of the Lemma~\ref{lem:cohenonline} algorithm applied to $S \intersect \{v_1, \ldots, v_i\}$.
Firstly, let $B_0 = \mathbf{0}_{d \times d}$, which is the initial condition for the algorithm in Lemma~\ref{lem:cohenonline} and is constructed by Algorithm~\ref{alg:sparsified-ks} on Line~\ref{algline:basecase}.
Then, for the inductive step, we assume that $B_{i-1}$ is the output of the Lemma~\ref{lem:cohenonline} algorithm applied to $S \intersect \{v_1, \ldots, v_{i-1}\}$ and we define $B_i$ by case distinction.

\textbf{Case 1: $v_i \not\in S$.}
In this case, we set $B_i = B_{i-1}$, and notice that if $B_{i-1}$ is in the set $\mathcal{L}_{i-1}$ constructed by Algorithm~\ref{alg:sparsified-ks}, then $B_i$ must be in the set $\mathcal{L}_i$ since every matrix $B$ in $\mathcal{L}_{i-1}$ is included in $\mathcal{L}_i$ on either Line~\ref{algline:lupdate} or Line~\ref{algline:lupdate2}.
Since $S \intersect \{v_1, \ldots, v_{i-1}\} = S \intersect \{v_1, \ldots, v_i\}$, we have that $B_i$ is the output of the algorithm defined in Lemma~\ref{lem:cohenonline} applied to $S \intersect \{v_1, \ldots, v_i\}$ by the inductive hypothesis.

\textbf{Case 2: $v_i \in S$.}
In this case, we set $B_i$ to be either $B_{i-1}$ or $B_{i-1} + (1/p) v_i v_i^\transpose$, according to the result of the condition on Line~\ref{algline:randif} of Algorithm~\ref{alg:sparsified-ks}.
Notice that, since the definition of $p$ in Algorithm~\ref{alg:sparsified-ks} is the same as the definition in Lemma~\ref{lem:cohenonline}, $B_i$ corresponds to the result of applying an iteration of the algorithm in Lemma~\ref{lem:cohenonline} with $B_{i-i}$ and $v_i$.
Therefore, by the induction hypothesis, $B_i$ is equivalent to the output of the Lemma~\ref{lem:cohenonline} algorithm applied to $S \intersect \{v_1, \ldots, v_i\}$, which completes the inductive construction of $B_1, \ldots, B_m$.

Finally, since our defined $B_m$ corresponds to the output of the algorithm in Lemma~\ref{lem:cohenonline} applied to $S$, we can apply Lemma~\ref{lem:cohenonline} to $S$ and $B_m$ which completes the proof.
\end{proof}

Finally, to prove Theorem~\ref{thm:algorithm}, we need only apply Lemma~\ref{lem:ssparsifier}
for the target set $S \subset \mathcal{I}$, and Lemma~\ref{lem:algpairs} for one of the pairs $(S', B)$ constructed by the algorithm.
In particular, we do not need to take the union bound over all sparsifiers constructed by the algorithm; rather,
it is enough that an accurate sparsifier is constructed for one specific target set.

\begin{proof}[Proof of Theorem~\ref{thm:algorithm}]
We first look at the case in which there is some $S \subset \mathcal{I}$,
such that for
\[
    A_S = \sum_{i \in S} v_i v_i^\transpose
\]
it holds that
\[
\left(\frac{1}{2} - c \sqrt{\alpha}\right) \leq x^\transpose A_S x \leq  \left(\frac{1}{2} + c \sqrt{\alpha} \right),
\]
for all unit vectors $x \in \R^d$.
By Lemma ~\ref{lem:ssparsifier}, with probability greater than or equal to $1-1/d$, there exists some pair $(S',B)\in \mathcal{L}_m$ such that
\begin{align}
    (1 - \mu) A_{S} - \delta I \preceq B \preceq (1 + \mu) A_{S} + \delta I, \label{eqn:thm1proof2}
\end{align}
where $\mu = \epsilon / 6$ and $\delta = \mu \min(c \sqrt{\alpha}, 1/2 - c \sqrt{\alpha}) \leq \mu$.
By Lemma~\ref{lem:algpairs}, with probability $1 - 1/d$, we also have that
\[
(1 - \mu) A_{S'} - \delta I \preceq B \preceq (1 + \mu) A_{S'} + \delta I,
\]
where $S'$ is the set constructed alongside $B$.
Taking the union bound over these two events, with probability at least $1-2/d$, we have for any unit vector $x\in \R^d$, that
\begin{align*}
      x^\transpose A_{S'} x & \le \frac{1+\mu}{1-\mu}\left(\frac{1}{2} + c \sqrt{\alpha}\right) + \frac{2 \delta}{1 - \mu} & & & x^\transpose A_{S'} x & \geq \frac{1-\mu}{1+\mu}\left(\frac{1}{2} - c \sqrt{\alpha}\right) - \frac{2 \delta}{1 - \mu}\\
    & \le  \frac{1+3\mu}{1-\mu}\left(\frac{1}{2} + c \sqrt{\alpha}\right) & \text{and} & &  & \ge  \frac{1-3\mu}{1-\mu}\left(\frac{1}{2} - c \sqrt{\alpha}\right)\\
    & \le (1+\epsilon)\left(\frac{1}{2} + c \sqrt{\alpha}\right) & & & & \ge  (1-\epsilon)\left(\frac{1}{2} - c \sqrt{\alpha}\right),
\end{align*}
where we use the definition of $\delta$ and the fact that $\epsilon = 6 \mu \leq 1$. 
Therefore, the set $S'$ satisfies \eqref{eq:KS_ep_condition} and will be returned by Algorithm~\ref{alg:sparsified-ks}.

On the other side, notice that, by the condition on Line~\ref{algline:if} of Algorithm~\ref{alg:sparsified-ks}, any set returned by the algorithm must satisfy  \eqref{eq:KS_ep_condition}. Therefore,
with probability $1$ the algorithm will correctly report that there is no set $S\subset \calI$ satisfying  \eqref{eq:KS_ep_condition} if it is the case.

Finally, we analyse the running time of the algorithm. By Lemma~\ref{lem:cohenonline}, it holds that $B$ is constructed from  $O(n)$ vectors with probability at least $1-1/d$.
For this reason, on Line~\ref{algline:loop} of Algorithm~\ref{alg:sparsified-ks} we consider only 
the sparsifiers of  size $\bigo{n}$. The remaining part of the algorithm contributes only polynomial factors to its running time, so the total running time of the algorithm is
\[
\bigo{\binom{m}{n}\cdot \mathrm{poly}(m, d)}. \qedhere
\]
\end{proof}

%


\section{$\FNP$-Hardness of $\mathsf{KS}_2\left(1/(4\sqrt{2})\right)$ \label{sec:hardness}}
This section studies the computational complexity of the $\KSc$ problem. We prove   that $\mathsf{KS}_2\left(1/(4\sqrt{2})\right)$ is $\FNP$-hard. 
This section is organised as follows. 
In Section~\ref{sec:tfnp} we introduce the $\mathsf{FNP}$ complexity class.   
We formally define the   \reductionsatnew\ problem in Section~\ref{sec:sat_npc}, and prove that this problem is  $\NP$-hard. In Section~\ref{sec:ksnpc}, we build a reduction from the \reductionsatnew\ problem to the $\mathsf{KS}_2(1/(4\sqrt{2}))$ problem. 

\subsection{The $\FNP$ Complexity Class \label{sec:tfnp}}
 
In contrast with the complexity classes $\mathsf{P}$ and $\NP$, the class $\FNP$ is used to study problems with output which is more complex than simply ``yes'' or ``no''.
Formally, given a binary relation $R$ and an input $X$, the corresponding \emph{function problem} is to find $Y$ such that $R(X, Y)$ holds or report ``no'' if no such $Y$ exists.
For example, we can take $X$ to be an instance $\mathcal{I} = \{v_i\}_{i = 1}^m$ of the $\KSc$ problem, and $Y \subseteq \mathcal{I}$ to be a candidate solution.
Then, the relation $R_{\KSc}(\mathcal{I}, Y)$ holds if and only if $Y$ satisfies \eqref{eq:KS_condition}.
Any given binary relation $R$ is in the class $\FNP$ 
iff there is a deterministic polynomial-time algorithm which can determine whether $R(X, Y)$ holds for a given pair $(X, Y)$~\cite{rich2008automata}.
Notice that every function problem has a natural corresponding decision problem.
Specifically, given a binary relation $R$ and a value of $X$, the decision problem asks whether there exists some $Y$ such that $R(X, Y)$ holds.
A function problem $F$ is $\FNP$-hard if there is a polynomial-time reduction from all problems in $\FNP$ to $F$.
It is known that if the decision problem corresponding to $F$ is $\NP$-hard, then $F$ is $\FNP$-hard~\cite{rich2008automata}, and we will use this fact in our proof of Theorem~\ref{thm:hardness}.

\subsection{$\mathsf{NP}$-Completeness of 
\reductionsatnew\ \label{sec:sat_npc}}

In this subsection we study the following \reductionsatnew\ problem, and prove that the problem is $\NP$-complete. We remark that we restrict ourselves to study SAT instances of  a specific form here, since these SAT instances will  be   employed to   prove the $\NP$-hardness of the $\mathsf{KS}_2\left(1/(4\sqrt{2})\right)$ problem.

\begin{problem}[\reductionsatnew] Given a 3SAT instance $\psi$ that consists of a collection $C$ of clauses over the set $U$ of variables such that 
\begin{enumerate}
    \item every clause $c\in C$ has $3$ literals,
    \item  for every $u\in U$, both of $u$ and $\bar{u}$ appear in at most $2$ clauses of $C$,
    \item  for every $u \in U$, at least one of $u$ or $\bar{u}$ appears in exactly $2$ clauses of $C$, and
    \item any two clauses share at most one literal and no variable appears twice in the same clause,
\end{enumerate}
 the \reductionsatnew\ problem asks if there is a satisfying assignment for $\psi$ such that every clause of $\psi$ has at least one true literal and at least one false literal. \label{prob:satnew}
\end{problem}
Our reduction is from the following well-known $\NP$-complete problem.
\begin{problem}[\reductionsatold, \cite{garey_computers_1979}]
Given a 3SAT instance $\psi$ that consists of a collection $C$ of clauses over the set $U$ of variables such that every clause $c\in C$ has $3$ literals, the \reductionsatold\ problem asks if there is a satisfying assignment for $\psi$ such that every clause of $\psi$ has at least one true   literal and at least one false literal.   
\end{problem}

\begin{theorem} \label{thm:nae2lit}
The \reductionsatnew\ problem is $\NP$-complete.
\end{theorem}
\begin{proof}
Given any \reductionsatnew\ instance $\psi$ and an assignment to $\psi$'s variables, it's straightforward to check in polynomial time if this is a satisfying assignment, and every clause of $\psi$ has at least one true literal and at least one false literal. Hence, the \reductionsatnew\ problem is in $\NP$.

To prove that the \reductionsatnew\ problem is $\NP$-complete, we build a reduction from the \reductionsatold\ problem to the \reductionsatnew\ problem. Specifically, for any \reductionsatold\ instance $(U,C)$, where $U$ is the set of variables and $C$ is a collection of clauses,  we construct an \reductionsatnew\ instance $(U',C')$ such that $(U,C)$ is satisfiable in \reductionsatold\ if and only if  $(U',C')$ is satisfiable in \reductionsatnew. Our construction of $(U',C')$ is as follows. Initially, we set $U'=U$ and $C'=C$.  
 Then, for any variable $x$ which appears only once in $C$, we remove $x$ from $U'$ and the corresponding clause from $C'$ since the clause can always be satisfied by setting $x$ appropriately and so removing the clause does not change the satisfiability of $(U', C')$.
Then, for every remaining variable $x$, we replace the instances of $x$ and $\bar{x}$ with new variables and add additional clauses to ensure that the satisfiability is unchanged.
Specifically, for each $x$ left in $U'$ let
\begin{itemize}
    \item $n_1 = \cardinality{\{c \in C : x \in c\}}$
    \item $n_2 = \cardinality{\{c \in C : \bar{x} \in c\}}$
\end{itemize}
 and set $n = n_1 + n_2$.
Then, we introduce new variables $x_1, \ldots, x_n$ and replace the instances of $x$ in $C'$ with $x_1, \ldots, x_{n_1}$.
Similarly, we replace the instances of $\bar{x}$ with $\bar{x}_{n_1 + 1}, \ldots, \bar{x}_{n}$.

Now, in order to ensure that $(U', C')$ is satisfiable if and only if $(U, C)$ is satisfiable, we   introduce new clauses to $C'$ which have the effect of constraining the variables $x_1, \ldots, x_n$ to have the same truth value in any satisfying assignment.
To achieve this, let $n' \geq n$ be an odd number, and we  introduce additional new variables $y_1, \ldots, y_{n'}$ and clauses
\begin{equation} \label{eq:yvars}
    \left(\bar{y}_{i} \lor \bar{y}_{i + 1} \lor y_{i + 2}\right)\quad\mbox{for any}\  i \in [1, n'], 
\end{equation}
where the  indices are taken modulo $n'$.
We will see that these clauses ensure that the $y_i$ variables must all have the same   value in a satisfying assignment.
We see this by a simple case distinction.
\begin{itemize}
    \item Case 1: $y_1 = y_2$ in a satisfying assignment. Then, by the first clause in \eqref{eq:yvars} it must be that $y_2 = y_3$ since there must be at least one true literal and one false literal in each satisfied clause. Proceeding inductively through the clauses in \eqref{eq:yvars}, we establish that $y_1 = y_2 = \ldots = y_{n'}$.
    \item Case 2: $y_1 \neq y_2$ in a satisfying assignment. We will show that this leads to a contradiction. By the last clause in \eqref{eq:yvars}, $y_{n'} \neq y_1$ since there must be at least one true literal and one false literal. Again, we proceed inductively from the $(n'-1)$th clause in \eqref{eq:yvars} down to establish that $y_1 \neq y_2, y_2 \neq y_3, \ldots,  y_{n'-1} \neq y_{n'}$. As such, we have $y_1 = y_3 = \ldots = y_{2i+1}$ which is a contradition since $n'$ is odd and we have already established that $y_1 \neq y_{n'}$.
\end{itemize}
As such, we can use the variables $y_1, \ldots, y_{n'}$ with the assumption that they have the same value in any satisfying assignment of $(U', C')$.
It remains to construct clauses to guarantee that the variables $x_1, \ldots, x_n$ have the same value in any satisfying assignment.
We add the clauses
\begin{equation} \label{eq:xvars}
    \left(x_{i} \lor \bar{x}_{i + 1} \lor y_i\right)\quad\mbox{for any}\  i \in [1, n], 
\end{equation}
where the   indices are taken modulo $n$.
We will show that $x_1 = x_2 = \ldots = x_n$ in a satisfying assignment by case distinction.
\begin{itemize}
    \item Case 1: $x_1 = y_i$ for all $i$. By the first clause in \eqref{eq:xvars}, it must be that $x_1 = x_2$ since we cannot have $x_1 = \bar{x}_2 = y_i$ in a satisfying assignment. Then, proceeding inductively using each clause in turn we establish that $x_1 = x_2 = \ldots = x_n$.
    \item Case 2: $\bar{x}_1 = y_i$ for all $i$. By the last clause in \eqref{eq:xvars}, it must be that $\bar{x}_n = \bar{x}_1$ since we cannot have $x_n = \bar{x}_1 = y_n$ in a satisfying assignment. Then, proceeding inductively from the $(n-1)$th clause down, we establish that $\bar{x}_1 = \bar{x}_2 = \ldots = \bar{x}_n$.
\end{itemize}
Notice that by this construction, each literal $x_i$, $\bar{x}_i$, $y_i$, and $\bar{y}_i$ now appears at most twice in $C'$, no two clauses share more than one literal and no literal appears twice in the same clause.
 Additionally, every $x_i$ and $\bar{x_i}$ appears exactly once in the clauses added by \eqref{eq:xvars}.
Since the variable $x_i$ also appears exactly once in the clauses corresponding directly to $C$, 
 requirement (3) of the \reductionsatnew\ problem is satisfied.
Moreover, we have that $(U', C')$ has a satisfying assignment if $(U, C)$ has a satisfying assignment; this follows by setting the values of  $x_1, \ldots, x_n$ in $U'$ to the value of their corresponding  $x \in U$.
On the other hand, any satisfying assignment of $(U', C')$ corresponds to a satisfying assignment of $(U, C)$, since we must have that $x_1 = \ldots = x_n$ and can set the value of  $x \in U$ to be the same value to get a satisfying assignment of $(U, C)$.
Finally, notice that our new instance $(U',C')$ of \reductionsatnew\ can be constructed in polynomial time in the size of the instance $(U,C)$ of \reductionsatold. 
This completes the proof. 
\end{proof} 
 
\subsection{$\FNP$--Hardness of $\mathsf{KS}_2\left(1/(4\sqrt{2})\right)$ \label{sec:ksnpc}}
We now show that the $\KSc$ problem is $\FNP$-hard for any $c\leq 1/(4\sqrt{2})$, i.e.,~Theorem~\ref{thm:hardness}. 
At a high level, the proof is by reduction from the \reductionsatnew\ problem.
Given an instance of the \reductionsatnew\ problem, we will construct an instance $\mathcal{I}$ of $\KSc$ such that
\begin{itemize}
\item if the \reductionsatnew\ instance is satisfiable, then there is a set $S \subset \mathcal{I}$ with $\sum_{v \in S} v v^\transpose = (1/2) \cdot I$, and
\item if the \reductionsatnew\ instance is not satisfiable, then for all sets $S \subset \mathcal{I}$ we have
 \[
    \norm{\sum_{v \in S} v v^\transpose - \frac{1}{2} I} \geq \frac{1}{4 \sqrt{2}} \sqrt{\alpha}.
\]
\end{itemize}
This will establish that the $\KScNP$ problem is $\FNP$-complete, and that it is $\NP$-hard to distinguish between instances of $\KSc$ with $\mathcal{W}(\mathcal{I}) = 0$ and those for which $\mathcal{W}(\mathcal{I}) \geq \left(1 / 4 \sqrt{2}\right) \sqrt{\alpha}$.

\begin{proof}[Proof of Theorem~\ref{thm:hardness}]
We prove that $\mathsf{KS}_2(1/(4\sqrt{2}))$ is $\NP$-hard by a reduction from the \reductionsatnew\ problem to the decision version of the $\KScNP$ problem. 
We are given an instance $(U, C)$ of the \reductionsatnew\ problem, and construct an instance of $\KSc$.
Let us refer to 
\begin{itemize}
    \item the clauses in $C$ as $c_1, \ldots c_m$;
    \item the variables in $U$ as $x_1, \ldots, x_n$ and we sometimes write $x_i$ and $\bar{x}_i$ for the un-negated and negated literals.
\end{itemize}
Our   constructed $\KSc$ instance has  $\bigo{n + m}$ dimensions.
Specifically, there is one dimension for each clause in $C$ and one dimension for each variable in $U$ which appears both negated and un-negated in $C$.
We   use
\begin{itemize}
    \item $d^c_j$ to refer to the dimension corresponding to clause $c_j$, and
    \item $d^x_j$ to refer to the dimension corresponding to variable $x_j$.
\end{itemize}
We   add $\bigo{m + n}$ vectors to our $\KSc$ instance.
Conceptually, we   add one vector for each clause and $4$ vectors for each literal.
We   use
\begin{itemize}
    \item $v^c_j$ to refer to the vector corresponding to clause $c_j$, and
    \item $v^x_{j, 1}$ to $v^x_{j, 4}$ or $v^{\bar{x}}_{j, 1}$ to $v^{\bar{x}}_{j, 4}$ to refer to the vectors corresponding to the literal $x_j$ or $\bar{x}_j$.
\end{itemize}
For each clause $c_j$, we set $v^c_j(d^c_j) = 1/2$, and set the other entries of $v^c_j$ to be $0$.
Table~\ref{tab:2literalvectors} completes the definition of the vectors corresponding to literals.
For each literal, we define only the value on the dimensions corresponding to the variable and the clauses containing the literal; all other entries in the vector are $0$. Let $A$ be the set of vectors defined above. Notice that the squared norms of the vectors in $A$ are bounded above by $1 / 4$ and so $\alpha=1 / 4$ in the constructed \textsf{KS}$_2(c)$ instance.

\begin{table}[h]
    \centering
    \begin{tabular}{cccc}
    \toprule
        Vector & Value on $d^c_j$ & Value on $d^c_k$ & Value on $d^x_i$  \\
        \midrule
        $v^x_{i, 1}$ & $1/4$ & $1/4$ & $1 / \sqrt{8}$ \\
        $v^x_{i, 2}$ & $1/4$ & $1/4$ & $- 1 / \sqrt{8}$ \\
        $v^x_{i, 3}$ & $1/4$ & $- 1/4$ & $1 / \sqrt{8}$ \\
        $v^x_{i, 4}$ & $1/4$ & $- 1/4$ & $- 1 / \sqrt{8}$ \\
    \bottomrule
    \end{tabular}
    \caption{The construction of the vectors in the $\KSc$ instance for a literal $x_i$ which appears in clause $c_j$ and possibly also in clause $c_k$. If a literal appears in only one clause, $c_j$, we ignore the middle column corresponding to $c_k$; that is, the vectors corresponding to $x_i$ are non-zero only on dimensions $d^c_j$ and $d^x_i$.}
    \label{tab:2literalvectors}
\end{table}

To complete the reduction, we will show the following:
\begin{enumerate}
    \item It holds that 
    \[
    \sum_{v \in A} v v^\transpose = I.
    \]
    \item If the original \reductionsatnew\ instance has a satisfying assignment, then there is a set $S \subset A$ such that
    \[
        \sum_{v \in S} v v^\transpose = \frac{1}{2} I.
    \]
    \item Any set $S \subset A$ such that
    \[
        \norm{\sum_{v \in S} v v^\transpose - \frac{1}{2} I} < \frac{1}{8 \sqrt{2}} = \frac{1}{4\sqrt 2}\sqrt \alpha
    \]
    corresponds to a satisfying assignment of the original \reductionsatnew\ instance.
\end{enumerate}
 
\paragraph{Vectors in $A$ are isotropic.}
Let us prove that
\[
    \sum_{v \in A} v v^\transpose = I.
\]
Let $B = \sum_{v \in A} v v^\transpose$. Then, for any variable $x_i$, we have that
\begin{align*}
    B(d^{x}_i, d^x_i) & = \sum_{v \in A} v(d^x_i)^2 \\
    & = \sum_{j=1}^4 v^x_{i, j}(d^x_i)^2 + \sum_{j=1}^4 v^{\bar{x}}_{i, j}(d^x_i)^2 \\
    & = 1.
\end{align*}
Additionally, for any clause $c_i$ we have that
\begin{align*}
    B(d^c_i, d^c_i) & = \sum_{v \in A} v(d^c_i)^2 \\
    & = v^c_i(d^c_i)^2 + \sum_{x_j \in c} \sum_{k=1}^4 v^x_{j, k}(d^c_i)^2 \\
    & = \frac{1}{4} + 3 \cdot \frac{1}{4} \\
    & = 1.
\end{align*}
This demonstrates that the diagonal entries of $B$ are all $1$.
We   now see that the off-diagonal entries are all $0$.
First, notice that for any two dimensions relating to variables, $d^x_i$ and $d^x_j$, we have
\begin{align*}
    B(d^x_i, d^x_j) & = \sum_{v \in A} v(d^x_i) v(d^x_j) = 0,
\end{align*}
since there is no vector in $A$ with a non-zero contribution to more than one dimension corresponding to a variable.
Now, let us consider two dimensions corresponding to different clauses $c_i$ and $c_j$.
We have
\begin{align*}
    B(d^c_i, d^c_j) & = \sum_{v \in A} v(d^c_i) v(d^c_j) \\
    & = \sum_{x_k \in c_i \cap c_j} \sum_{\ell = 1}^4 v^x_{k, \ell}(d^c_i) \cdot v^x_{k, \ell}(d^c_j) \\
    & = \sum_{x_k \in c_i \cap c_j} \frac{1}{16} + \frac{1}{16} - \frac{1}{16} - \frac{1}{16} \\
    & = 0,
\end{align*}
where we use the fact that $c_i$ and $c_j$ share at most one literal.
Finally, consider the case when one dimension corresponds to the clause $c_i$ and the other dimension corresponds to the variable $x_j$.
If the variable $x_j$ does not appear in $c_i$, then there are no vectors with a non-zero contribution to the two dimensions and so the entry is $0$.
Otherwise, we have
\begin{align*}
    B(d^c_i, d^x_j) & = \sum_{v \in A} v(d^c_i) v(d^x_j) \\
    & = \sum_{k = 1}^4 v^x_{i, k}(d^c_i) v^x_{i, k}(d^x_j) \\
    & = \frac{1}{4 \sqrt{8}} + \frac{1}{4 \sqrt{8}} - \frac{1}{4 \sqrt{8}} - \frac{1}{4 \sqrt{8}} \\
    & = 0,
\end{align*}
where we use the fact that no variable appears twice in the same clause.
This completes the proof that
\[
    \sum_{v \in A}v v^\transpose = I.
\]

\paragraph{If the \reductionsatnew\ instance is satisfiable, then there is a solution to $\mathsf{KS}_2(1/(4\sqrt{2}))$.}
Given a satisfying assignment to the \reductionsatnew\ problem, let $T \subset U$ be the set of variables which are set to be \textsc{True} and let $F \subset U$ be the set of  variables which are set to be \textsc{False}.
Recall that in a satisfying assignment, each clause in $C$   contains either $1$ or $2$ true literals.
Let $C' \subset C$ be the set of clauses with exactly $1$ true literal in the satisfying assignment.
Then, we define $S$ to be
\[
    S \triangleq \{ v^x_{i, 1}, v^x_{i, 2}, v^x_{i, 3}, v^x_{i, 4} : x_i \in T \} \union 
    \{ v^{\bar{x}}_{i, 1}, v^{\bar{x}}_{i, 2}, v^{\bar{x}}_{i, 3}, v^{\bar{x}}_{i, 4} : x_i \in F \} \union
    \{ v^c_i : c_i \in C' \}.
\]
and we   show that
\[
    \sum_{v \in S} v v^\transpose = \frac{1}{2} I.
\]
Now, we can repeat the calculations of the previous paragraph, this time setting $B = \sum_{v \in S} v v^\transpose$ to show that $B = (1/2) I$.
Specifically, for any variable $x_i$, it holds that
\begin{align*}
    B(d^x_i, d^x_i) = \frac{1}{2}
\end{align*}
since only the vectors corresponding to the negated \emph{or} un-negated variable are included.
For any clause $c_i \in C'$, we have
\begin{align*}
    B(d^c_i, d^c_i) & = v^c_i(d^c_i)^2 + \sum_{k=1}^4 v^x_{j, k}(d^c_i)^2 \\
    & = \frac{1}{4} + \frac{1}{4} = \frac{1}{2}
,\end{align*}
where $x_j$ is the literal which is set to be true in the clause $c_i$.
Similarly, for any clause in $c_i \in C \setminus C'$, we have
\begin{align*}
    B(d^c_i, d^c_i) & = \sum_{k=1}^4 v^x_{j, k}(d^c_i)^2 + \sum_{k=1}^4 v^x_{\ell, k}(d^c_i)^2 \\
    & = \frac{1}{4} + \frac{1}{4} = \frac{1}{2},
\end{align*}
where the literals $x_j$ and $x_{\ell}$ are set to be true in the clause $c_i$.
Then, notice that the calculations for the off-diagonal entries follow in the same way as before.
This completes the proof that a satisfying assignment for the \reductionsatnew\ problem implies a solution to the $\mathsf{KS}_2(1/(4\sqrt{2}))$ problem.  

\paragraph{If there is a solution to $\KSc$, then the \reductionsatnew\ instance is satisfiable.} 
We prove this by a contrapositive argument. That is, we show   that for any set $S'$ which does not correspond to a satisfying assignment of the \reductionsatnew\ problem, there must be some vector $y$ with $\norm{y} = 1$ such that
\begin{equation} \label{eq:epserror}
    \abs{y^\transpose \left(\sum_{v \in S'} v v^\transpose\right) y - \frac{1}{2}} \geq \epsilon
\end{equation}
for $\epsilon = \frac{1}{8 \sqrt{2}}$.
Specifically, we will analyse three cases, and show that

\begin{enumerate}
    \item if there is some variable $x_i$ such that $S'$ does not contain exactly $4$ of the vectors \[\{v^x_{i, 1}, v^x_{i, 2}, v^x_{i, 3}, v^x_{i, 4}, v^{\bar{x}}_{i, 1}, v^{\bar{x}}_{i, 2}, v^{\bar{x}}_{i, 3}, v^{\bar{x}}_{i, 4}\},\] then there is a vector $y$ satisfying \eqref{eq:epserror} for $\epsilon = 1/8$;
    \item if Item~(1) does not apply, then if there is some literal $x_i$ such that $S'$ contains $1$, $2$, or $3$ of the vectors $\{v^x_{i, 1}, v^x_{i, 2}, v^x_{i, 3}, v^x_{i, 4}\}$, then there is a vector $y$ satisfying \eqref{eq:epserror} for $\epsilon = \frac{1}{8 \sqrt{2}}$;
    \item if neither Item~(1) nor (2) applies, then if $S'$ does not correspond to a satisfying assignment of the original \reductionsatnew\ instance, there must be a vector $y$ satisfying \eqref{eq:epserror} for $\epsilon = 1/4$.
\end{enumerate}

For the first case, suppose that there is some variable $x_i$ such that $S'$ does not contain exactly $4$ vectors corresponding to the variable $x_i$.
Let $k \neq 4$ be the number of such vectors, and let $y$ be the vector with all zeros except for $y(d^x_i) = 1$. Notice that
\begin{align*}
    \abs{y^\transpose \left(\sum_{v \in S'} v v^\transpose\right) y - \frac{1}{2}}
    = \abs{ \sum_{v \in S'} v(d^x_i)^2 - \frac{1}{2} }
    = \abs{ \frac{k}{8} - \frac{1}{2} }
    \geq \frac{1}{8}.
\end{align*}

For the second case, suppose that the set $S'$ contains $4$ vectors for each variable, but there is some literal $x_i$ such that $S'$ contains some but not all of the vectors corresponding to $x_i$.
By condition $3$ of the \reductionsatnew\ problem (Problem~\ref{prob:satnew}) we can assume that $x_i$ appears in two clauses $c_j$ and $c_k$. Otherwise, this is the case for $\bar{x}_i$ and $S'$ contains some, but not all, of the vectors corresponding to $\bar{x}_i$ since it contains exactly $4$ vectors corresponding to the variable $x_i$.
Now, we define
\[
    B = \sum_{v \in S'} v v^\transpose
\]
and we   consider the absolute values of certain off-diagonal entries in $B$, which are summarised in Table~\ref{tab:offdiags}.
\begin{table}[t]
    \centering
    \begin{tabular}{cccc}
    \toprule
        Vectors in $S'$ & $ |B(d^x_i, d^c_j)|$ & $\abs{B\left(d^x_i, d^c_k\right)}$ & $ |B(d^c_j, d^c_k)|$ \\
        \midrule
        One vector $v^x_{i, \ell}$ & $1 / \left(8 \sqrt{2}\right)$ & $1 / \left(8 \sqrt{2}\right)$ & $1 / 16$ \\
        Vectors $v^x_{i, 1}$ and $v^x_{i, 2}$ & $0$ & $0$ & $1 / 8$ \\
        Vectors $v^x_{i, 1}$ and $v^x_{i, 3}$ & $1 / \left(4 \sqrt{2}\right)$ & $0$ & $0$ \\
        Vectors $v^x_{i, 1}$ and $v^x_{i, 4}$ & $0$ & $1 / \left( 4 \sqrt{2}\right)$ & $0$ \\
        Vectors $v^x_{i, 2}$ and $v^x_{i, 3}$ & $0$ & $1 / \left(4 \sqrt{2}\right)$ & $0$ \\
        Vectors $v^x_{i, 2}$ and $v^x_{i, 4}$ & $1 / \left(4 \sqrt{2}\right)$ & $0$ & $0$ \\
        Vectors $v^x_{i, 3}$ and $v^x_{i, 4}$ & $0$ & $0$ & $1 / 8$ \\
        Three vectors $v^x_{i, \ell}$ & $1 /\left(8 \sqrt{2}\right)$ & $1 /\left(8 \sqrt{2}\right)$ & $1 / 16$ \\
    \bottomrule
    \end{tabular}
    \caption{The absolute values of certain off-diagonal entries in $B = \sum_{v \in S'} v v^\transpose$, depending on which vectors corresponding to the literal $x_i$ are included in $S'$. We assume that $x_i$ appears in the clauses $c_j$ and $c_k$.}
    \label{tab:offdiags}
\end{table}
Notice that, regardless of which vectors corresponding to $x_i$ are included, there are two indices $\hat{d}_1$ and $\hat{d}_2$ such that
\[
    \abs{B(\hat{d}_1, \hat{d}_2)} \geq \frac{1}{8 \sqrt{2}}.
\]
Using the indices $\hat{d_1}$ and $\hat{d_2}$, define the unit vector 
\begin{align*}
y = \twopartdefow
    {\frac 1 {\sqrt 2}(\mathbf{1}_{\hat{d_1}}+\mathbf{1}_{\hat{d_2}})}
    {\sgn(B(\hat{d_1},\hat{d_1})+B(\hat{d_2},\hat{d_2})-1) = \sgn(B(\hat{d_1},\hat{d_2}))}
    {\frac 1 {\sqrt 2}(\mathbf{1}_{\hat{d_1}}-\mathbf{1}_{\hat{d_2}})}
\end{align*}
where $\sgn(\cdot)$ is the sign function.
Then we have
\begin{align*}
    \abs{y^{\rot}By - \frac 1 2} &= \abs{\frac 1 2 \Big(B(\hat{d_1},\hat{d_1}) + B(\hat{d_2},\hat{d_2}) \pm B(\hat{d_1},\hat{d_2}) \pm B(\hat{d_2},\hat{d_1})\Big)-\frac 1 2}\\
    & = \frac 1 2\abs{B(\hat{d_1},\hat{d_1}) + B(\hat{d_2},\hat{d_2}) -1 \pm 2B(\hat{d_1},\hat{d_2})}\\
    &=\frac 1 2 \left(\abs{B(\hat{d_1},\hat{d_1})+B(\hat{d_2},\hat{d_2}) -1} + 2\abs{B(\hat{d_1},\hat{d_2})}\right)\\
    & \geq \abs{B\left(\hat{d}_1, \hat{d}_2\right)} \\
    & \geq \frac{1}{8 \sqrt{2}},
\end{align*}
where the third equality follows by the construction of $y$.\\

Finally, we consider the third case, in which there are $4$ vectors in $S'$ for each variable, and all $4$ vectors correspond to the same literal.
It is clear that such a set $S'$ corresponds unambiguously to an assignment for the original variables in the \reductionsatnew\ instance: specifically, one can set a variable $x_i$ to be \textsc{True} if $S'$ contains $\{v^x_{i, 1}, v^x_{i, 2}, v^x_{i, 3}, v^x_{i, 4}\}$, and  set $x_i$ to be \textsc{False} if $S'$ contains $\{v^{\bar{x}}_{i, 1}, v^{\bar{x}}_{i, 2}, v^{\bar{x}}_{i, 3}, v^{\bar{x}}_{i, 4}\}$.
Then, suppose that there is some clause $c_j \in C$ which is not satisfied by this assignment.
This implies that  either all $12$ of the vectors corresponding to literals in $c_j$ are included in $S'$, or none of the vectors corresponding to literals in $c_j$ are included in $S'$.
In either case, we can  set $y$ to be the indicator vector of the dimension $d^c_j$, and have that
\[
    \abs{y^\transpose B y - \frac{1}{2}} = \abs{\sum_{v \in S'} v(d^c_j)^2 - \frac{1}{2}} \geq \frac{1}{4}
\]
since we can either include $v^c_j$ or not in order to set $\sum_{v \in S'} v(d^c_j)^2$ equal to either $1/4$ or $3/4$.
 
This completes the reduction from the \reductionsatnew\ problem to the decision version of the $\KSc$ problem for $c \leq 1 / (4 \sqrt{2})$ which implies that $\KScNP$ is $\FNP$-hard. 
Furthermore, notice that by the reduction in this proof,
\begin{itemize}
    \item if the \reductionsatnew\ instance is satisfiable, then the constructed instance $\mathcal{I}$ of the $\KSc$ problem satisfies $\mathcal{W}(\mathcal{I}) = 0$;
    \item if the \reductionsatnew\ instance is not satisfiable, then the constructed instance $\mathcal{I}$ of the $\KSc$ problem satisfies $\mathcal{W}(\mathcal{I}) \geq 1 / \left(4 \sqrt{2}\right) \cdot \sqrt{\alpha}$.
\end{itemize}
This shows that distinguishing between instances with $\mathcal{W}(\mathcal{I}) = 0$ and $\mathcal{W}(\mathcal{I}) \geq 1 / \left(4 \sqrt{2} \right) \cdot \sqrt{\alpha}$ is $\NP$-hard, and completes the proof.
\end{proof}

\section{Conclusion}
This paper studies the algorithms and complexity of the Kadison-Singer problem through the $\KSc$ problem, and presents two results.
On one side, we prove that the $\KSc$ problem for any $c\in\mathbb{R}^+$ can be solved in quasi-polynomial time when $d=O(\log m)$, which  suggests  that the   problem is much easier to solve in low dimensions.
The key  to our algorithm design is a novel application of online spectral
sparsification subroutines, with which we are able to efficiently construct representations of all spectral
equivalence classes over time and reduce the enumeration space of the candidate solutions.
We expect that our work could motivate more research on the applications of spectral sparsification and related problems in numerical linear algebra to the algorithmic Kadison-Singer problem.


On the other side, our \textsf{NP}-hardness result shows that the Kadison-Singer type problem for arbitrary dimensions can be as hard as solving the \textsf{SAT} problem,  and the $\KSc$ problem belongs to different complexity classes for different values of $c$. Hence, more refined studies on the classification of its computational complexity would help us better understand the complexity of  the algorithmic Kadison-Singer problem.
In our point of view, both directions left from the paper are very interesting, and we leave these for future work.

\paragraph{Acknowledgement.} We would like to thank an anonymous reviewer for their detailed and valuable comments on   earlier versions of our paper.  These comments helped us significantly improve the presentation of the paper.

\bibliographystyle{alpha}
\bibliography{references.bib}

\newpage

\appendix

\end{document}